\documentclass[12pt]{amsart}
\usepackage{amsmath}
\usepackage{amssymb}
\usepackage{mathrsfs}
\usepackage[top=3cm,bottom=3.5cm,right=2.5cm,left=2.5cm]{geometry}
\newcommand{\Z}{\mathbb{Z}}
\newcommand{\var}{{\rm var}}
\newcommand{\R}{\mathbb{R}}
\newcommand{\N}{\mathbb{N}}
\newcommand{\diam}{{\rm diam}}
\newcommand{\I}{\mathscr{I}}
\newcommand{\M}{\mathscr{M}}
\newcommand{\interior}{{\rm int}}
\newcommand{\E}{\mathscr{E}}
\newcommand{\F}{\mathscr{F}}

\usepackage{amsthm}
\newtheorem{dfn}{Definition}[section]
\newtheorem{prop}{Proposition}[section]
\newtheorem{thm}[prop]{Theorem}

\newtheorem{lem}[prop]{Lemma}
\newtheorem{prob}[prop]{Problem}
\theoremstyle{definition}
\newtheorem{ex}{Example}[section]
\newtheorem{rem}{Remark}[section]
\usepackage{mathtools}
\DeclarePairedDelimiter{\abs}{\lvert}{\rvert}
\makeatletter
\newcommand{\opnorm}{\@ifstar\@opnorms\@opnorm}
\newcommand{\@opnorms}[1]{%
	\left|\mkern-1.5mu\left|\mkern-1.5mu\left|
	#1
	\right|\mkern-1.5mu\right|\mkern-1.5mu\right|
}
\newcommand{\@opnorm}[2][]{%
	\mathopen{#1|\mkern-1.5mu#1|\mkern-1.5mu#1|}
	#2
	\mathclose{#1|\mkern-1.5mu#1|\mkern-1.5mu#1|}
}
\makeatother
\makeatletter

\@addtoreset{equation}{section}

\makeatother

\title[stability of uniqueness and coexistence of equilibrium states]{Stability of Uniqueness and Coexistence of Equilibrium States of the Ising Model under Long Range Perturbations}
\author{Shunsuke Usuki}
\date{}

\begin{document}
	\maketitle

	\begin{abstract}
		In this paper, we study perturbations of the $d$-dimensional Ising model for $d\geq 2$, including long range ones to which the Pirogov-Sinai theory is not applicable.
		We show that the uniqueness of the equilibrium state of the Ising model at high temperature and the coexistence of equilibrium states at low temperature are preserved by spin-flip symmetric perturbations.
	\end{abstract}

	
	\section{Introduction}\label{intro}
	In this paper, we deal with {\bf $\Z^d$-lattice systems} for $d\geq 2$. We begin with the definition of them.
	Let $F$ be a finite set and $\Omega=F^{\Z^d}$. Then $\Omega$ is a compact metrizable space with respect to the product topology. For $\Lambda\subset\Z^d$, we write $\Omega_\Lambda=F^\Lambda$ and $\omega=\left\{\omega(x)\left|x\in\Lambda\right.\right\}$ as an element of $\Omega_\Lambda$. 
	For each $a\in\Z^d$, the {\bf translation} by $a$ is canonically defined, that is, for each $\Lambda\subset\Z^d$, we define $\tau^a:\Omega_\Lambda\to\Omega_{\Lambda-a}$ by
	$$(\tau^a\omega)(x)=\omega(x+a),\quad x\in \Lambda-a$$
	for $\omega\in\Omega_\Lambda$.
	When $\Lambda=\Z^d$, these translations define the $\Z^d$-action on $\Omega$.
	In this paper, we write $\Lambda\Subset\Z^d$ when $\Lambda\subset\Z^d$ and $\Lambda$ is a finite set.
	We say $\Phi$ is an {\bf interaction} if $\Phi$ is a family of functions $\Phi_\Lambda:\Omega_\Lambda\to\R$  for each $\Lambda\Subset\Z^d$ and write $\Phi=\left\{\Phi_\Lambda\right\}_{\Lambda\Subset\Z^d}$. 
	We always assume in this paper that an interaction $\Phi$ is {\bf translation invariant}, that is, for any $a \in {\Z}^d$ and $\Lambda\Subset {\Z}^d$,
	$$\Phi_{\Lambda}(\omega)=\Phi_{\Lambda-a}(\tau^a\omega).$$
	We also assume that $\Phi$ is {\bf absolutely summable}, that is, the norm 
	$$\|\Phi\|=\sum_{0\in\Lambda\Subset{\Z}^d}\sup_{\xi\in \Omega_{\Lambda}}\left| \Phi_{\Lambda}(\xi)\right|$$
	is finite.
	Under such an interaction $\Phi$, for each $\Lambda\Subset\Z^d$ and $\omega\in\Omega$, the {\bf Hamiltonian} is defined by
	\begin{equation*}
		H_{\Phi,\Lambda}(\omega)=\sum_{\Delta\in\F_\Lambda}\Phi_\Delta(\omega),
	\end{equation*}
	where $\F_\Lambda=\left\{\Delta\Subset\Z^d\left|\Delta\cap\Lambda\neq\emptyset\right.\right\}$ and $\Phi_\Delta(\omega)=\Phi_\Delta\left(\{\omega(x)|x\in\Delta\}\right)$. By the translation invariance and the absolute summability of $\Phi$, it is easily seen that the infinite sum converges.
	Then $\Omega$ is interpreted as the {\bf configuration space} and $H_{\Phi,\Lambda}\ (\Lambda\Subset\Z^d)$ define the energy for each configuration $\omega\in\Omega$.
	Then we regard $\Omega$ with $\Phi$ as a $\Z^d$-lattice system.
	We say that an interaction $\Phi$ is {\bf finite range} if
	$$
	r(\Phi)=\inf\left\{R>0\left| \Phi_\Lambda=0\text{ for any }\Lambda\Subset\Z^d \text{ with }\diam(\Lambda)>R\right.\right\}<\infty.
	$$
	In the above, for $\Lambda\subset\Z^d$,
	$$\diam(\Lambda)=\sup\left\{\left.\|x-y\|_\infty\right| x,y\in\Lambda\right\},$$
	where $\|x\|_\infty=\max\left\{|x_1|,\ldots,|x_d|\right\}$ for $x=(x_1,\ldots,x_d)\in\Z^d$.
	We emphasize that, in this paper, we deal with interactions which are {\bf long range}, that is, $r(\Phi)=\infty$ in general.
	
	The {\bf Ising model} is a famous and important example of a $\Z^d$-lattice system, which we mainly deal with in this paper. This is a simple mathematical model of ferromagnetism.
	\begin{dfn}[The Ising model]
		\label{Ising}
		Let $\beta>0$. The $d$-dimensional Ising model at the inverse temperature $\beta$
		is a $\Z^d$-lattice system on $\Omega=\{1,-1\}^{\Z^d}$ with a finite range interaction $\Phi^\beta$ defined by
		$$\Phi^\beta_{\Lambda}(\omega)=\begin{cases}
			-\beta\omega(x)\omega(y) , &\text{if $\Lambda=\{x,y\}$ with $\|x-y\|_1=1$}\\
			0,&\text{otherwise,} \end{cases}
		$$
		for each $\Lambda\Subset\Z^d$ and $\omega\in\Omega_{\Lambda}=\{1,-1\}^{\Lambda}$, where $\|x\|_1=\sum_{i=1}^{d}|x_i|$ for $x=(x_1,\ldots,x_d)\in\Z^d$. Here, $\beta$ is the parameter representing the inverse of the temperature of the system.
	\end{dfn}
	
	For a $\Z^d$-lattice system, Borel probability measures on $\Omega$ which represent
	statistical states of the system stable under the interaction are determined and called {\bf Gibbs states}.
	Translation invariant probability measures representing stable
	statistical states are called {\bf equilibrium states}. We give the rigorous definition in Section \ref{latticesystems}.
	It is an important problem to determine whether there are more than one equilibrium state for the system or not. For the case of $d=1$, it is known that there exists the only one equilibrium state for any interaction under a reasonable condition. For example, finite range interactions satisfy this condition (see \cite[Chapter 5]{Ru}). On the other hand, the system might have more than one equilibrium state when $d\geq2$. Concerning this problem, the Ising model defined in Definition \ref{Ising} shows the remarkable property below (for example, see \cite[Chapter 3]{FV1}).
	\begin{prop}[Phase transition on the Ising model]
		\label{phasetransIsing}
		Let $d\geq2$. Then there exists $\beta_c>0$ depending on $d$ such that the followings hold. 
		\begin{enumerate}
			\renewcommand{\labelenumi}{\rm{(\arabic{enumi})}}
			\item If $\beta<\beta_c$, then $\Phi^\beta$ has a unique Gibbs state and, in particular, a unique equilibrium state. 
			\item If $\beta>\beta_c$, then $\Phi^\beta$ has more than one equilibrium state.
		\end{enumerate}
	\end{prop}
	This phenomenon is called {\bf phase transition} on the Ising model and $\beta_c$ as above the critical inverse temperature.
	\begin{rem}
		\begin{enumerate}
			\renewcommand{\labelenumi}{\rm{(\roman{enumi})}}
			\item For the Ising models with $\beta>\beta_c$, the concrete structure of the set of equilibrium states is known (see \cite{Aiz80} or \cite{Hig81} for $d=2$ and \cite{Bod06} for $d\geq 3$). We will mention this structure in detail later (Section \ref{furtherquestions}).
			\item
			It seems difficult to study the Ising model at the critical temperature. However, it is shown in \cite{ADS15} that, at $\beta=\beta_c$, $\Phi^{\beta_c}$ has a unique Gibbs state.
		\end{enumerate}
	\end{rem}
	
	The central question of this paper is to determine for the Ising model whether the uniqueness or coexistence of equilibrium states is preserved if we perturb the interaction $\Phi^\beta$. 
	Here a perturbation of $\Phi^\beta$ means an interaction $\Phi^\beta+\Psi$ defined by
	$$
	(\Phi^\beta+\Psi)_\Lambda=\Phi^\beta_\Lambda+\Psi_\Lambda
	$$
	for each $\Lambda\Subset\Z^d$ for some small perturbation interaction $\Psi$.
	Perturbations of the interaction may correspond to some influence from the surroundings to the system or some noise in the interaction. Hence it is interesting to study what happens under perturbations of the interaction. 
	We remark that the `upper semicontinuity' of equilibrium states in perturbations of interactions holds in the following sense.
	\begin{prop}\label{perturbcont}
		Let $\Phi_0$ be an interaction, $I_{\Phi_0}$ be the set of equilibrium states for $\Phi_0$ and $I$ be the set of translation invariant Borel probability measures on $\Omega$.
		Then, for any neighborhood $U$ of $I_{\Phi_0}$ in $I$ with respect to the weak*-topology,
		there exists $0<\delta<1$ such that, for any interaction $\Psi$ with $\|\Psi\|<\delta$, we have
		$$I_{\Phi_0+\Psi}\subset U.$$
	\end{prop}
	We give the proof in the appendix. We notice that Proposition \ref{perturbcont} does not ensure the stability of the uniqueness or coexistence of equilibrium states.
	
	If the perturbations are finite range, then this perturbation problem can be studied by the known theory, called the {\bf Pirogov-Sinai theory} and it is known that the coexistence of equilibrium states at low temperature is stable under finite range perturbations which preserve some symmetry of the Ising model, called the spin-flip symmetry.
	
	\begin{dfn}
		\label{symm}
		An interaction $\Psi$ on $\Omega=\{1,-1\}^{\Z^d}$ is {\bf spin-flip symmetric} if
		$$\Psi_{\Lambda}(\omega)=\Psi_{\Lambda}(-\omega)$$
		for $\Lambda\Subset\Z^d$ and $\omega\in\Omega_\Lambda$.
	\end{dfn}
	
	Our main theorem says that the same statement holds under long range perturbations to which the Pirogov-Sinai theory is not applicable. We remark in detail the relation between our main theorem and the Pirogov-Sinai theory in Secrion \ref{pirogovsinai}.
	We consider perturbations by small parturbation interactions with respect to the norm $\opnorm{\cdot}$, which we call the {\bf $d$-th order decaying norm}. The following is our main theorem.
	
	\begin{thm}[Main theorem]
		\label{mainthm}
		Let $d\geq2$ and, for $\beta>0$, $\Phi^\beta$ be the interaction of the $d$-dimensional Ising model at the inverse temperature $\beta$. We write $\beta_c$ for the critical inverse temperature. Then we have the following.
		\begin{enumerate}
			\renewcommand{\labelenumi}{\rm{(\arabic{enumi})}}
			\item There exist $0<\beta_h<\beta_c$ and $\delta>0$ such that, for any $0<\beta\leq\beta_h$ and any translation invariant interaction $\Psi$ with $\opnorm{\Psi}<\delta$, $\Phi=\Phi^\beta+\Psi$ has a unique Gibbs state and, in particular, a unique equilibrium state.
			\item There exist $\beta_c<\beta_l<\infty$ and $\delta>0$ such that, for any $\beta\geq\beta_l$ and any translation invariant interaction $\Psi$ with $\opnorm{\Psi} <\delta$ which is spin-flip symmetric, $\Phi=\Phi^\beta+\Psi$ has more than one equilibrium state.
		\end{enumerate}
	\end{thm}
	
	The definition of $\opnorm{\cdot}$ is given in Section \ref{sectionmainthm}.
	We notice that this norm is stronger than $\|\cdot\|$ and weaker than the exponentially decaying norm.
	\begin{rem}\label{strongerresult}
		\begin{enumerate}
			\renewcommand{\labelenumi}{\rm{(\roman{enumi})}}
			\item Actually, Theorem \ref{mainthm} (1) is a direct collorary of {\bf Dobrushin's criterion}, which gives a condition for an interaction to have a unique Gibbs state. We remark this in detail in Section \ref{sectionDob}.
			\item Theorem \ref{mainthm} (2) follows from stronger results, Theorems \ref{mainthmstronger} and \ref{mainthmconnected} below.
			\item In Theorem \ref{mainthm}, the $d$-th order decaying condition is crucial. See Remark \ref{magnetization}.
		\end{enumerate}
	\end{rem}
	
	We notice that Theorem \ref{mainthm} (2) fails if we drop the spin-flip symmetry condition. To see this,
	we see the phase diagram of the {\bf Ising model with external fields}.
	For $\beta,h\in\R$, we define the interaction $\{\Phi^{\beta,h}_\Lambda\}_{\Lambda\Subset\Z^d}$ by
	$$
	\Phi^{\beta,h}_\Lambda(\omega)=
	\begin{cases}
		-\beta\omega(x)\omega(y),&\Lambda=\{x,y\}\text{ with }\|x-y\|_1=1\\
		-h\omega(x),&\Lambda=\{x\}\\
		0,&\text{otherwise},
	\end{cases}
	$$
	for each $\Lambda\Subset\Z^d$ and $\omega\in\Omega_\Lambda$. When $h=0$ then $\Phi^{\beta,0}=\Phi^\beta$ and, by taking $h\neq 0$ arbitrarily small, we can think $\Phi^{\beta,h}$ as a perturbation of $\Phi^\beta$.
	Then the following holds (see \cite[Chapter 3]{FV1}).
	\begin{prop}\label{phasetransIsingext}
		For any $\beta,h\in\R$ with $h\neq 0$, $\Phi^{\beta,h}$ has a unique Gibbs state.
	\end{prop}
	We notice that the perturbations $\Phi^{\beta,h}\ (h\neq0)$ are not spin-flip symmetric.
	Hence Proposition \ref{phasetransIsingext} says that, for sufficiently large $\beta$, non-spin-flip symmetric perturbations of $\Phi^\beta$ can break the structure of the equilibrium states.
	Moreover, it is known that the uniqueness of the equilibrium state is a `generic' property in the space of interactions (with appropriate norms). This fact is discussed in \cite{DE80}, \cite{Ent81}, \cite{Sok82} and \cite{Isr86}.
	Our main theorem \ref{mainthm} says that spin-flip symmetric perturbations are `not generic' and preserve the coexistence. 
	
	In Section \ref{latticesystems}, we give the rigorous definition of Gibbs states and equilibrium states. 
	We give the definition of the $d$-th order decaying norm $\opnorm{\cdot}$ is Section \ref{sectionmainthm}. In Section \ref{sectionDob}, we give Dobrushin's criterion and deduce Theorem \ref{mainthm} (1) from it.
	We remark the relation between our main theorem and the Pirogov-Sinai theory in Secrion \ref{pirogovsinai}.
	We prove the main theorem at low temperature in Section \ref{sectperturbIsinglow}. We use in the proof some extension of Peierls' argument.
	We notice further questions in Section \ref{furtherquestions}.
	
	The author is grateful to Masayuki Asaoka and Mitsuhiro Shishikura for their helpful advice, and Takehiko Morita, Masaki Tsukamoto and Tom Meyerovitch for their useful comments.
	
	\section{Preliminaries}
	\subsection{Gibbs states and equilibrium states}\label{latticesystems}
	In this section, for a $\Z^d$-lattice system on $\Omega=F^{\Z^d}$ with an interaction $\Phi$, we give the rigorous definition of Gibbs states and equilibrium states. We first define Gibbs states by the thermodynamical condition, called the DLR condition\footnote {This is named after R. L. Dobrushin, O. E. Lanford and D. Ruelle.}. We next define equilibrium states by the variational principle in the theory of dynamical systems and notice that these different definitions are equivalent under the translation invariant condition.
	
	For $\Lambda\Subset\Z^d, \eta\in\Omega_{\Z^d\setminus\Lambda}$ and $\xi_\Lambda\in\Omega_\Lambda$, we write $\xi_{\Lambda}\lor\eta$ for the element of $\Omega$ defined by  $\xi_{\Lambda}\lor\eta|_{\Lambda}=\zeta_{\Lambda}$ and $\xi_{\Lambda}\lor\eta|_{\Z^d\setminus\Lambda}=\eta$. We define the {\bf finite Gibbs state $\mu_{\Phi,\Lambda}^\eta$ in $\Lambda$ with the boundary condition $\eta$} as the Borel probability measure on $\Omega_\Lambda$ defined by
	\begin{equation}\label{Gibbscond}
		\mu_{\Phi,\Lambda}^\eta\left(\{\omega_{\Lambda}\}\right)=\frac{1}{Z_{\Phi,\Lambda}^\eta}\exp\left(-H_{\Phi,\Lambda}(\omega_\Lambda\lor\eta)\right)
	\end{equation}
	for each $\omega_\Lambda\in\Omega_\Lambda$, where
	\begin{equation*}
		Z_{\Phi,\Lambda}^\eta=\sum_{\xi_{\Lambda}\in\Omega_{\Lambda}}\exp\left(-H_{\Phi,\Lambda}(\xi_{\Lambda}\lor\eta)\right)\nonumber.
	\end{equation*}
	We call $Z_{\Phi,\Lambda}^\eta$ the {\bf partition function in $\Lambda$ with the boundary condition $\eta$}.
	Let $C(\Omega)$ be the real Banach space of continuous real-valued functions on $\Omega$ with the supremum norm and $M(\Omega)$ be the set of Borel probability measures on $\Omega$. $M(\Omega)$ can be viewed as a subset of the dual space of the Banach space $C(\Omega)$. Then $M(\Omega)$ is a compact, convex and metrizable space with respect to the weak*-topology.
	
	\begin{dfn}[Gibbs State]
		\label{Gibbs}
$\sigma\in M(\Omega)$ is a Gibbs state for an interaction $\Phi$ if for each $\Lambda\Subset\Z^d$, there exists a Borel probability measure $\sigma_{\Z^d\setminus\Lambda}$ on $\Omega_{\Z^d\setminus\Lambda}$ such that for any $\omega_{\Lambda}\in\Omega_{\Lambda}$,
		$$\sigma(\{\xi\in\Omega\left|\ \xi|_{\Lambda}=\omega_{\Lambda}\right \})=\int_{\Omega_{\Z^d\setminus\Lambda}}\mu_{\Phi,\Lambda}^\eta\left(\{\omega_{\Lambda}\}\right)d\sigma_{\Z^d\setminus\Lambda}(\eta).$$
	\end{dfn}
	Definition \ref{Gibbs} is equivalent to the following: for each $\Lambda\Subset\Z^d$ and $\eta\in\Omega_{\Z^d\setminus\Lambda}$, the conditional probability that $\xi|_{\Lambda}=\omega_{\Lambda}$ under $\xi|_{\Z^d\setminus\Lambda}=\eta$ is $\mu_{\Phi,\Lambda}^\eta\left(\{\omega_{\Lambda}\}\right)$. This condition is called the DLR condition.
	Let $K_{\Phi}$ be the set of Gibbs states for $\Phi$. This is a nonempty, compact and convex subset of $M(\Omega)$ (see \cite[Chapter 1]{Ru}). 
	We notice that $K_{\Phi+C}=K_\Phi$ for any constant interaction $C$ (that is, an interaction such that $C_\Lambda$ is constant for each $\Lambda\Subset\Z^d$).
	
	Next, we state the variational principle in the theory of dynamical systems and give the definition of equilibrium states.
	Let $I$ be the set of Borel probability measures $\mu$ on $\Omega$ which are translation invariant, that is, $\tau^a_*\mu=\mu$ for each $a\in\Z^d$.
	This is a nonempty, compact and convex subset of $M(\Omega)$. For a sequence $\{\Lambda_n\}_{n=1}^{\infty}$ of finite subsets of $\Z^d$, we write $\Lambda_n\nearrow\infty$ (limit in the sense of van Hove) when $|\Lambda_n|\rightarrow\infty$ and, for any $a\in\Z^d$, $|(\Lambda_n+a)\setminus\Lambda_n|\big/|\Lambda_n|\rightarrow0$. An example of such a sequence is $\left\{B(n)\right\}_{n=1}^\infty$, where $B(n)=\{-n,\dots,n\}^d$.
	\begin{dfn}
		\label{entropy}
		For each $\mu\in I$, the limit
		$$h(\mu)=\lim_{n\rightarrow\infty}-\frac{1}{|\Lambda_{n}|}\sum_{\omega\in\Omega_{\Lambda_{n}}}\mu_{\Lambda_n}(\{\omega\})\log\mu_{\Lambda_n}(\{\omega\}),$$
		where $\mu_{\Lambda_n}(\{\omega\})=\mu(\{\xi\in\Omega\left|\  \xi|_{\Lambda_n}=\omega \right \})$, exists and is independent of the choice of sequence $\{\Lambda_n\}_{n=1}^{\infty}$ such that $\Lambda_n\nearrow\infty$ (see {\rm \cite[Chapter 3]{Ru}}). $h(\mu)$ is called the {\bf measure-theoretic entropy} of $\mu$.
	\end{dfn}
	
	\begin{dfn}
		\label{pressure}
		For each interaction $\Phi$ and $\Lambda\Subset\Z^d$, We define the {\bf partition function in $\Lambda$ with free boundary condition} by
		$$
		Z_{\Phi,\Lambda}=\sum_{\omega\in\Omega_{\Lambda}}\exp\left(-\sum_{\Delta\subset\Lambda}\Phi_\Delta(\omega)\right).
		$$
		Then, the limit
		$$P^{\Phi}=\lim_{n\rightarrow\infty}\frac{1}{|\Lambda_{n}|}\log Z_{\Phi,\Lambda_n}$$
		exists and is independent of the choice of sequence $\{\Lambda_n\}_{n=1}^{\infty}$ such that $\Lambda_n\nearrow\infty$ (see {\rm\cite[Chapter 3]{Ru}}). $P^{\Phi}$ is called the {\bf pressure} of $\Phi$.
	\end{dfn}
	
	For an interaction $\Phi$, we define $A_{\Phi}:\Omega\to\R$ by
	$$A_{\Phi}(\omega)=-\sum_{\Lambda}\Phi_{\Lambda}(\omega),$$
	where the sum runs over those $\Lambda\Subset\Z^d$ such that $\Lambda$ contains $0$ in $\Z^d$ and $0$ is the middle element of $\Lambda$ (that is, the $\lfloor \left(|\Lambda|+1\right)/2\rfloor$-th\footnote{For $a\in\R$, $\lfloor a\rfloor$ denotes the largest integer that is not larger than $a$.} element) with respect to the lexicographic order on $\Lambda$. Since $\Phi$ is absolutely summable, $A_\Phi$ is continuous. $-A_{\Phi}(\omega)$ physically represents the contribution of $0\in\Z^d$ to the energy in the configuration $\omega$. The {\bf variational principle} is the statement which connencts entropy, pressure and $A_\Phi$.
	
	\begin{prop}[Variational Principle]
		\label{VP}
		For each interaction $\Phi$,
		$$P^{\Phi}=\sup_{\mu\in I}\left(h(\mu)+\int_{\Omega}A_{\Phi}d\mu\right).$$
		Moreover, there exists some $\mu\in I$ such that $\mu$ achieves the supremum (see {\rm \cite[Theorem 3.12]{Ru}}).
	\end{prop}
	
	\begin{dfn}[Equilibrium State]
		An element $\mu\in I$ is an equilibrium state for an interaction $\Phi$ if $\mu$ achieves the supremum of Proposition \ref{VP}.
	\end{dfn}
	Let $I_{\Phi}$ be the set of the equilibrium states for $\Phi$. This is a nonempty, compact and convex subset of $M(\Omega)$. The definituon of the equilibrium states turn out to be a charactarization of the translation invariant Gibbs states, that is, the following proposition holds (see \cite[Theorem 4.2]{Ru}).
	\begin{prop}
		\label{GE}
		For each interaction $\Phi$,
		$$I_{\Phi}=K_{\Phi}\cap I.$$
	\end{prop}
	
	\subsection{$d$-th order decaying interactions}\label{sectionmainthm}
	In this section, we introduce $d$-th order decaying interactions on $\Omega=\{-1.1\}^{\Z^d}$ and give the definition of the $d$-th order decaying norm $\opnorm{\cdot}$ in Theorem \ref{mainthm}.
	
	\begin{dfn}
		\label{decreasing}
		An interaction $\Psi$ is {\bf $d$-th order decaying} if
		$$\opnorm{\Psi}=\sum_{0\in\Lambda\Subset\Z^d}(\diam(\Lambda)+1)^d\sup_{\omega\in\Omega_\Lambda}\left|\Psi_\Lambda(\omega)\right|<\infty.$$
		We call the norm $\opnorm{\cdot}$ the $d$-th order decaying norm.
	\end{dfn}
	It is obvious that
	\begin{equation*}
		\|\Psi\|\leq\opnorm{\Psi}
	\end{equation*}
	for any $d$-th order decaying interaction $\Psi$ on $\Omega$. We see some examples of $d$-th order decaying interactions. Clearly, we have many $d$-th order decaying interactions other than these examples.
	\begin{ex}
		\begin{enumerate}
			\renewcommand{\labelenumi}{\rm{(\arabic{enumi})}}
			\item
			Finite range interactions are clearly $d$-th order decaying.
			\item(Two body interactions.)
			Let $\Psi$ be a translation invariant interaction on $\Omega$ which satisfies
			$$\Psi_{\Lambda}=0$$
			unless $\Lambda$ consists of distinct two points. We emphasize that these two body interactions are not finite range in general.
			We assume that $\Psi$ is $(2d+\varepsilon)$-th order decaying, that is, there exists $\varepsilon>0$ such that
			$$\|\Psi\|_\varepsilon=\sup\left\{\left.\|x-y\|_\infty^{2d+\varepsilon}\left|\Psi_{\{x,y\}}(\omega)\right|\ \right|\ x,y\in\Z^d,\ x\neq y,\ \omega\in\Omega_{\{x,y\}}\right\}<\infty.$$
			We have
			\begin{align*}
				\opnorm{\Psi}=&\sum_{x\in\Z^d\setminus\{0\}}(\|x\|_\infty+1)^d\sup_{\omega\in\Omega_{\{0,x\}}}\left|\Psi_{\{0,x\}}(\omega)\right|\\
				\leq&\sum_{x\in\Z^d\setminus\{0\}}(\|x\|_\infty+1)^d\|x\|_\infty^{-(2d+\varepsilon)}\|\Psi\|_\varepsilon\\
				=&M_\varepsilon\|\Psi\|_\varepsilon,
			\end{align*}
			where $ M_\varepsilon=\sum_{n=1}^\infty((2n+1)^d-(2n-1)^d)(n+1)^dn^{-(2d+\varepsilon)}<\infty$. Hence $\Psi$ is $d$-th order decaying.
		\end{enumerate}
	\end{ex}
	Here, we give some remark on two body interactions. If a two body interaction $\Psi$ has a form
	$\Psi_{\{x,y\}}(\omega)=K_{x,y}\omega(x)\omega(y)$, $K_{x,y}\in\R$,
	a condition for $\Psi$ to have more than one equilibrium state is given in \cite{GGR66}. From this condition, we can see that, 
	if a perturbation interaction $\Psi$ is a two body interaction which has a form as above, the coexistence of equilibrium states holds under weaker assumption than Theorem \ref{mainthm} (2).
	
	\subsection{High temperature cases and Dobrushin's criterion}\label{sectionDob}
	In this section, we give Dobrushin's criterion and see that Theorem \ref{mainthm} (1) follows from this.
	Let $\Phi$ be a translation invariant interaction. For any $x\in\Z^d\setminus\{0\}$, we define
	$$\rho_{\Phi}(x)=\sup\left\{\left|\mu_{\Phi,\{0\}}^\eta(\{1\})-\mu_{\Phi,\{0\}}^\zeta(\{1\})\right|\left|\  \eta, \zeta\in\Omega_{\Z^d\setminus\{0\}},\  \eta|_{\Z^d\setminus\{0,x\}}=\zeta|_{\Z^d\setminus\{0,x\}} \right. \right\}.$$
	\begin{prop}[Dobrushin's Criterion]
		\label{Dob}
		If
		$$\sum_{x\in\Z^d\setminus\{0\}}\rho_{\Phi}(x)<1,$$
		then $\Phi$ has a unique Gibbs state.
	\end{prop}
	The following is a convenient version of Dobrushis's criterion. 
	We define
	$$\var_{\Lambda}\Phi=\sup_{\omega, \omega'\in\Omega_{\Lambda}}\left|\Phi_{\Lambda}(\omega)-\Phi_{\Lambda}(\omega')\right|$$
	for each $\Lambda\Subset\Z^d$ and
	$$\|\Phi\|_\var=\sum_{0\in\Lambda\Subset\Z^d}(|\Lambda|-1)\var_{\Lambda}\Phi,$$
	where $|\Lambda|$ be the cardinality of $\Lambda$.
	It is obvious that
	\begin{equation}
		\label{comparison}
		\|\Phi\|_\var\leq2\opnorm{\Phi}
	\end{equation}
	for any $d$-th order decaying interaction $\Phi$.
	\begin{prop}
		\label{Dobapp}
		If
		$$\|\Phi\|_\var<2,$$
		then $\Phi$ satisfies Dobrushin's criterion and hence has a unique Gibbs state.
	\end{prop}
	We refer \cite[Chapter 8]{Geo11} for proofs of the above propositions. It is easy to see that $\|\Phi^\beta\|_\var=4d\beta$. From this, Inequality (\ref{comparison}) and Proposition \ref{Dobapp}, we have Theorem \ref{mainthm} (1).
	
	\subsection{Low temperature cases and Pirogov-Sinai theory}\label{pirogovsinai}
	
	We notice the relation between Theorem \ref{mainthm} (2) and the Pirogov-Sinai theory.
	If $\beta$ is sufficiently large and {\bf a perturbation interaction $\Psi$ is finite range}, then the equilibrium states of $\Phi=\Phi^\beta+\Psi$ can be studied by an application of the Pirogov-Sinai theory. It was introduced by S. A. Pirogov and Ya. G. Sinai in \cite{PS} and sophisticated in \cite{Zah84} to study phase diagrams of more general finite range interactions at low temperature. We refer \cite[Chapter 7]{FV1} for an intelligible introduction of this theory. The periodic ground states of $\Phi^\beta$, that is, the states which have the lowest energy under $\Phi^\beta$, are exactly the constant configurations $\eta^+$ and $\eta^-\in\Omega$, which is constant $1$ and $-1$ on $\Z^d$, respectively, and it can be seen that sufficiently small and spin-flip symmetric perturbations of $\Phi^\beta$ preserves $\eta^+$ and $\eta^-$ as periodic ground states. Hence, by the Pirogov-Sinai theory, if $\Psi$ is sufficiently small, spin-flip symmetric and finite range interaction, then $\Phi=\Phi^\beta+\Psi$ has equilibrium states corresponding to two periodic ground states, $\eta^+$ and $\eta^-$.
	
	The assumption that interactions are finite range is crucial in the Pirogov-Sinai theory. Hence Theorem \ref{mainthm} (2), including results in the case of long range perturbations, does not follows from it. The Pirogov-Sinai theory was extended to some extent in \cite{Par}, but this extention is possible only for sum of finite range interactions and two body interactions satisfying some exponentially decaying condition. 
	It was also extended to quantum cases in \cite{BKU96}, or a series of papers \cite{DFF96} and \cite{DFR96}. In these cases, perturbation interactions can be long range but must exponentially decay.
	Hence Theorem \ref{mainthm} (2), allowing $d$-th order decaying perturbations, also does not follow from these extended version.
	On the other hand, we can not study spin-flip symmetric perturbations (which are long range in general) in detail using the techniques in this paper as well as finite range perturbations using the Pirogov-Sinai theory. We will mention this later (Section \ref{furtherquestions}).
	
	\section{Perturbations of the Ising model at low temperature}\label{sectperturbIsinglow}
	
	\subsection{The main theorem in low temperature cases and reformalization of perturbation interactions}
	In Section \ref{sectperturbIsinglow}, we give a proof of Theorem \ref{mainthm} (2). As we mentioned in Remark \ref{strongerresult}, Theorem \ref{mainthm} (2) follows from stronger results: Theorems \ref{mainthmstronger} and \ref{mainthmconnected} below. We first state these results.
	
	We introduce a norm $\opnorm{\cdot}'$, which is weaker than the $d$-th order decaying norm $\opnorm{\cdot}$. For a translation invariant interaction $\Psi$, we define
	$$
	\opnorm{\Psi}'=\sum_{0\in\Lambda\Subset\Z^d}|\Lambda|^{-1}(\diam(\Lambda)+1)^d\sup_{\omega\in\Omega_\Lambda}\left|\Psi_\Lambda(\omega)\right|.
	$$
	It is clear that $\opnorm{\Phi}'\leq\opnorm{\Phi}$. Then we have the following theorem.
	\begin{thm}\label{mainthmstronger}
		Let $d\geq2$. Then there exists $0<L<\infty$ satisfying the following.
		If $\beta,\delta>0$ satisfy $\beta-\delta>L$, then,
		for any translation invariant interaction $\Psi$ with $\opnorm{\Psi}' <\delta$ which is spin-flip symmetric, $\Phi=\Phi^\beta+\Psi$ has more than one equilibrium state.
	\end{thm}
	Theorem \ref{mainthmstronger} follows from Theorem \ref{mainthmconnected} below. 
	We say that $\Lambda\subset\Z^d$ is {\bf $\|\cdot\|_1$-connected} if for each $x$ and $y\in\Lambda$ there exist finite points $x_0,x_1,\ldots,x_{n-1},x_n\in\Lambda$ such that
	$x_0=x,x_n=y,\|x_{i+1}-x_i\|_1=1$ for $i=0,\ldots,n-1$.
	We call such a finite sequence $x_0,x_1,\ldots ,x_{n-1},x_n$ a {\bf $\|\cdot\|_1$-path} from $x$ to $y$. For an interaction $\Psi$, we say that $\Psi$ is {\bf zero on non-$\|\cdot\|_1$-connected sets}
	if $\Psi_\Lambda\equiv0$ whenever $\Lambda\Subset\Z^d$ is not $\|\cdot\|_1$-connected.
	\begin{thm}\label{mainthmconnected}
		Let $d\geq2$. Then there exists $0<L<\infty$ satisfying the following.
		If $\beta,\delta>0$ satisfy $\beta-\delta>L$, then,
		for any translation invariant interaction $\Psi$ with $\|\Psi\|<\delta$ which is spin-flip symmetric and zero on non-$\|\cdot\|_1$-connected set, $\Phi=\Phi^\beta+\Psi$ has more than one equilibrium state.
	\end{thm}
	
	We see that how Theorem \ref{mainthmstronger} reduces to Theorem \ref{mainthmconnected}.
	For this purpose, we reformalize a perturbation interaction. Let $\Psi$ be a translation invariant interaction with $\opnorm{\Psi}'<\infty$.
	For $R\Subset\Z^d$, we say $R$ is a {\bf rectangle} if $R=[a_1,b_1]\times\cdots\times[a_d,b_d]\cap\Z^d$ for some $a_1\leq b_1,\ldots,a_d\leq b_d\in\Z$. For each finite $\Lambda\Subset\Z^d$, we write $R(\Lambda)$ for the minimal rectangle which contains $\Lambda$ with respect to the partial order given by inclusion. We notice that $\diam(\Lambda)=\diam(R(\Lambda))$. For each rectangle $R$, we define $\mathscr{S}(R)=\left\{\Lambda\Subset\Z^d\left| R(\Lambda)=R\right.\right\}$.
	For $\Psi$ as above, we define an interaction $\widetilde{\Psi}$ as
	$$\widetilde{\Psi}_\Lambda\equiv0$$
	if $\Lambda$ is not a rectangle and
	$$\widetilde{\Psi}_R(\omega)=\sum_{\Lambda\in\mathscr{S}(R)}\Psi_\Lambda(\omega|_\Lambda)$$
	for any rectangle $R\Subset\Z^d$ and $\omega\in\Omega_R$. This interaction $\widetilde{\Psi}$ is made by putting interactions $\Psi_\Lambda$ on $\Lambda\in\mathscr{S}(R)$ together into the rectangle $R$. Obviously $\widetilde{\Psi}$ is a translation invariant interaction. We write $\mathscr{S}_0$ for the set of $\Lambda\Subset\Z^d$ such that $\Lambda$ contains $0$ in $\Z^d$ and $0$ is the middle element of $\Lambda$ (that is, the $\left[ \left(\lfloor\Lambda\rfloor+1\right)/2\right]$-th element) with respect to the lexicographic order on $\Lambda$. Then, by the translation-invariance of $\widetilde{\Psi}$, we have
	\begin{align}\label{opnorm}
		\|\widetilde{\Psi}\|=&\sum_{0\in R\Subset\Z^d, R:{\rm rectangle}}\sup_{\omega\in\Omega_R}|\widetilde{\Psi}_R(\omega)| \nonumber\\
		=&\sum_{\substack{0\in R\Subset\Z^d,R:{\rm rectangle},\\R\in\mathscr{S}_0}}|R|\sup_{\omega\in\Omega_R}|\widetilde{\Psi}_R(\omega)| \nonumber\\
		\leq&\sum_{\substack{0\in R\Subset\Z^d,R:{\rm rectangle},\\R\in\mathscr{S}_0}}|R|\sum_{\Lambda\in\mathscr{S}(R)}\sup_{\omega\in\Omega_\Lambda}|\Psi_\Lambda(\omega)| \nonumber\\
		=&\sum_{\substack{\Lambda\Subset\Z^d,\\R(\Lambda)\in\mathscr{S}_0}}|R(\Lambda)|
		\sup_{\omega\in\Omega_\Lambda}|\Psi_\Lambda(\omega)| \nonumber\\
		\stackrel{(*)}{=}&\sum_{0\in\Lambda\Subset\Z^d}|\Lambda|^{-1}|R(\Lambda)|
		\sup_{\omega\in\Omega_\Lambda}|\Psi_\Lambda(\omega)| \nonumber\\
		\leq&\sum_{0\in\Lambda\Subset\Z^d}|\Lambda|^{-1}\left(\diam(\Lambda)+1\right)^d
		\sup_{\omega\in\Omega_\Lambda}|\Psi_\Lambda(\omega)| \nonumber\\
		=&\ \opnorm{\Psi}'.
	\end{align}
	Here the equation $(*)$ holds since $\Psi$ is translation invariant and, for each  $\Lambda\Subset\Z^d$, there is exactly one translation of $\Lambda$ that appears in the summation $\sum_{\Lambda\Subset\Z^d,R(\Lambda)\in\mathscr{S}_0}$. The following lemma shows that $\widetilde{\Psi}$ and $\Psi$ induces the same statistical systems.
	\begin{lem}\label{sameGibbs}
		Let $\Psi$ and $\widetilde{\Psi}$ be the interactions as above.
		For any interaction $\Phi_0$, interactions $\Phi=\Phi_0+\Psi$ and $\widetilde{\Phi}=\Phi_0+\widetilde{\Psi}$ have the same Gibbs states.
	\end{lem}
	
	In particular, by Proposition {\rm\ref{GE}}, $\Phi$ and $\widetilde{\Phi}$ have the same equilibrium states.
	
	\begin{proof}[{\rm Proof}]
		We take arbitrary $\Lambda\Subset\Z^d$ and $\eta\in\Omega_{\Z^d\setminus\Lambda}$. For any $\omega_\Lambda\in\Omega_\Lambda$, we have, by (\ref{Gibbscond}),
		\begin{align*}
			&\mu_{\widetilde{\Phi},\Lambda}^\eta(\{\omega_\Lambda\})\\
			=&\frac{\displaystyle\exp\left(-\sum_{\Delta\in\mathscr{F}_\Lambda}\left(\Phi_{\Delta}\left(\omega_{\Lambda}\lor\eta\right)+\widetilde{\Psi}_\Delta\left(\omega_\Lambda\lor\eta\right)\right)\right)}{\displaystyle\sum_{\xi_{\Lambda}\in\Omega_{\Lambda}}\exp\left(-\sum_{\Delta\in\mathscr{F}_\Lambda}\left(\Phi_\Delta\left(\xi_{\Lambda}\lor\eta\right)+\widetilde{\Psi}_\Delta\left(\xi_\Lambda\lor\eta\right)\right)\right)}\\
			=&\frac{\displaystyle\exp\left(-\sum_{\Delta\in\mathscr{F}_\Lambda}\Phi_\Delta\left(\omega_{\Lambda}\lor\eta\right)-\sum_{R\in\mathscr{R}_\Lambda}\sum_{S\in\mathscr{S}(R)}\Psi_S\left(\omega_\Lambda\lor\eta\right)\right)}{\displaystyle\sum_{\xi_{\Lambda}\in\Omega_{\Lambda}}\exp\left(-\sum_{\Delta\in\mathscr{F}_\Lambda}\Phi_\Delta\left(\xi_{\Lambda}\lor\eta\right)-\sum_{R\in\mathscr{R}_\Lambda}\sum_{S\in\mathscr{S}(R)}\Psi_S\left(\xi_\Lambda\lor\eta\right)\right)},
		\end{align*}
		where $\mathscr{R}_\Lambda=\left\{\left. R\Subset\Z^d \right|R\textrm{ is a rectangle and }R\cap\Lambda\neq\emptyset\right\}.$
		Here, for any $\xi_\Lambda\in\Omega_\Lambda$,
		\begin{align*}
			&\sum_{R\in\mathscr{R}_\Lambda}\sum_{S\in\mathscr{S}(R)}\Psi_S\left(\xi_\Lambda\lor\eta\right)\\
			=&\sum_{R\in\mathscr{R}_\Lambda}\sum_{S\in\mathscr{S}(R)\cap\mathscr{F}_\Lambda}\Psi_S\left(\xi_\Lambda\lor\eta\right)
			+\sum_{R\in\mathscr{R}_\Lambda}\sum_{S\in\mathscr{S}(R)\setminus\mathscr{F}_\Lambda}\Psi_S\left(\eta\right)\\
			=&\sum_{\Delta\in\mathscr{F}_\Lambda}\Psi_\Delta\left(\xi_\Lambda\lor\eta\right)
			+\sum_{R\in\mathscr{R}_\Lambda}\sum_{S\in\mathscr{S}(R)\setminus\mathscr{F}_\Lambda}\Psi_S\left(\eta\right)
		\end{align*}
		and the second term is independent of $\xi_\Lambda$. Hence we have
		\begin{align*}
			&\mu_{\widetilde{\Phi},\Lambda}^\eta(\{\omega_\Lambda\})\\
			=&\frac{\displaystyle\exp\left(-\sum_{\Delta\in\mathscr{F}_\Lambda}\Phi_\Delta\left(\omega_{\Lambda}\lor\eta\right)-\sum_{R\in\mathscr{R}_\Lambda}\sum_{S\in\mathscr{S}(R)}\Psi_S\left(\omega_\Lambda\lor\eta\right)\right)}{\displaystyle\sum_{\xi_{\Lambda}\in\Omega_{\Lambda}}\exp\left(-\sum_{\Delta\in\mathscr{F}_\Lambda}\Phi_\Delta\left(\xi_{\Lambda}\lor\eta\right)-\sum_{R\in\mathscr{R}_\Lambda}\sum_{S\in\mathscr{S}(R)}\Psi_S\left(\xi_\Lambda\lor\eta\right)\right)}\\
			=&\frac{\displaystyle\exp\left(-\sum_{\Delta\in\mathscr{F}_\Lambda}\left(\Phi_\Delta\left(\omega_{\Lambda}\lor\eta\right)+\Psi_\Delta\left(\omega_\Lambda\lor\eta\right)\right)\right)}{\displaystyle\sum_{\xi_{\Lambda}\in\Omega_{\Lambda}}\exp\left(-\sum_{\Delta\in\mathscr{F}_\Lambda}\left(\Phi_\Delta\left(\xi_{\Lambda}\lor\eta\right)+\Psi_\Delta\left(\xi_\Lambda\lor\eta\right)\right)\right)}\quad\\
			=&\ \mu_{\Phi,\Lambda}^\eta(\{\omega_\Lambda\}).
		\end{align*}
		This implies that each finite Gibbs state coinsides for $\Phi$ and $\widetilde{\Phi}$. Hence the statement holds.
	\end{proof}
	Obviously $\widetilde{\Psi}$ is zero on non-$\|\cdot\|_1$-connected sets. Hence, by replacing $\Psi$ with $\widetilde{\Psi}$ in Theorem \ref{mainthmstronger} and using (\ref{opnorm}) and Lemma \ref{sameGibbs}, we see that Theorem \ref{mainthmstronger} follows from Theorem \ref{mainthmconnected}.
	
	\subsection{Gibbs states with the $+,-$-boundary condition}
	We will give a proof of Theorem \ref{mainthmconnected}.
	Let $\Phi$ be a translation invariant interaction such that $\|\Phi\|<\infty$.
	For $\Lambda\Subset\Z^d$, we define the {\bf finite Gibbs state with the $+$} and {\bf $-$-boundary condition}: $\mu_{\Phi,\Lambda}^+$ and $\mu_{\Phi,\Lambda}^-$ as the Borel probability measures on $\Omega_\Lambda$ defined by (\ref{Gibbscond}) with $\eta=\eta_{\Z^d\setminus\Lambda}^+$ and $\eta_{\Z^d\setminus\Lambda}^-$: the constant $1$ and $-1$ configurations on $\Z^d\setminus\Lambda$, respectively.
	That is, 
	\begin{align}\label{plusminusgibbs}
		&\mu_{\Phi,\Lambda}^+\left(\{\omega\}\right)=\frac{1}{Z_{\Phi,\Lambda}^+}\exp\left(-\sum_{\Delta\in\mathscr{F}_\Lambda}\Phi_\Delta(\omega\lor\eta_{\Z^d\setminus\Lambda}^+)\right)\text{ and }\nonumber\\
		&\mu_{\Phi,\Lambda}^-\left(\{\omega\}\right)=\frac{1}{Z_{\Phi,\Lambda}^-}\exp\left(-\sum_{\Delta\in\mathscr{F}_\Lambda}\Phi_\Delta(\omega\lor\eta_{\Z^d\setminus\Lambda}^-)\right)
	\end{align}
	for each $\omega\in\Omega_\Lambda$, where 
	\begin{align*}
		&Z_{\Phi,\Lambda}^+=\sum_{\omega\in\Omega_\Lambda}\exp\left(-\sum_{\Delta\in\mathscr{F}_\Lambda}\Phi_\Delta(\omega\lor\eta_{\Z^d\setminus\Lambda}^+)\right)\text{ and }\\
		&Z_{\Phi,\Lambda}^-=\sum_{\omega\in\Omega_\Lambda}\exp\left(-\sum_{\Delta\in\mathscr{F}_\Lambda}\Phi_\Delta(\omega\lor\eta_{\Z^d\setminus\Lambda}^-)\right).
	\end{align*}
	
	Here, we state the key proposition: Proposition \ref{lowtempmain} and see Theorem \ref{mainthmconnected} follows from it.
	We say that $\Lambda\subset\Z^d$ is {\bf $\|\cdot\|_\infty$-connected} if for each $x$ and $y\in\Lambda$ there exist finite points $x_0,x_1,\ldots,x_{n-1},x_n\in\Lambda$ such that
	$x_0=x,x_n=y,\|x_{i+1}-x_i\|_\infty=1$ for $i=0,\ldots,n-1$.
	We call such a finite sequence $x_0,x_1,\ldots ,x_{n-1},x_n$ a {\bf $\|\cdot\|_\infty$-path} from $x$ to $y$. Moreover, we say that $\Lambda$ is {\bf c-connected} if $\Lambda$ and $\Z^d\setminus\Lambda$ are $\|\cdot\|_\infty$-connected.
	\begin{prop}\label{lowtempmain}
		Let $d\geq2$. Then there exists sufficiently large $0<L<\infty$ and $0<\varepsilon(L)<1/2$ with $\varepsilon(L)\to0$ as $L\to\infty$ satisfying the following.
		Let $\beta,\delta>0$ and $\beta-\delta>L$ and $\Psi$ be a translation invariant interaction with $\|\Psi\|<\delta$ which is spin-flip symmetric and zero on non-$\|\cdot\|_1$-connected sets. We write $\Phi=\Phi^\beta+\Psi$. Then, for any c-connected $\Lambda\Subset\Z^d$ and $x\in\Lambda$, we have
		$$\mu_{\Phi,\Lambda}^+\left(\left\{\omega\in\Omega_\Lambda\left|\omega(x)=1\right.\right\}\right)>1-\varepsilon(L) \quad\text{and}\quad\mu_{\Phi,\Lambda}^-\left(\left\{\omega\in\Omega_\Lambda\left|\omega(x)=-1\right.\right\}\right)>1-\varepsilon(L).$$
	\end{prop}
	This proposition means that, with respect to $\mu_{\Phi,\Lambda}^+$ and $\mu_{\Phi,\Lambda}^-$, the spin at $x$ is magnetized for each $x\in\Lambda$. Let us see that Theorem \ref{mainthmconnected} follows from Proposition \ref{lowtempmain}.
	For a sequence $\left\{\Lambda_n\right\}_{n=1}^\infty$ of finite subsets of $\Z^d$, we weite $\Lambda_n\uparrow \Z^d$ when $\Lambda_n\subset\Lambda_{n+1}$ for each $n$ and $\bigcup_{n=1}^\infty\Lambda_n=\Z^d$. For example, $\Lambda_n=B(n)$ satisfies these conditions. For each $\Lambda\Subset\Z^d$, we canonically regard $\mu_{\Phi,\Lambda}^+$ and $\mu_{\Phi,\Lambda}^-$ as Borel probability measures on $\Omega$ with $\mu_{\Phi,\Lambda}^+(\Omega_{\Lambda}^+)=1$ and $\mu_{\Phi,\Lambda}^-(\Omega_{\Lambda}^-)=1$, where $\Omega_{\Lambda}^+=\left\{\omega\in\Omega\left|\omega(x)=1,x\in\Z^d\setminus\Lambda\right.\right\}$ and $\Omega_{\Lambda}^-=\left\{\omega\in\Omega\left|\omega(x)=-1,x\in\Z^d\setminus\Lambda\right.\right\}$, respectively. It is known that if $\mu_{\Phi,\Lambda_n}^+$ (resp. $\mu_{\Phi,\Lambda_n}^-$) converges to some $\mu^+$ (resp. $\mu^-$) in $M(\Omega)$ as $n\to\infty$, then $\mu^+\in K_\Phi$ (resp. $\mu^-\in K_\Phi$) (see \cite[Chapter 1]{Ru}).
	\begin{proof}[{\rm Proof of Theorem {\rm \ref{mainthmconnected}}}]
		Assume that Proposition \ref{lowtempmain}  holds. We take $L,\beta,\delta,\Psi$ and $\Phi$ as in Proposition \ref{lowtempmain}.
		We take a sequence $\left\{\Lambda_n\right\}_{n=1}^\infty$ of finite and c-connected subsets of $\Z^d$ such that $\Lambda_n\uparrow\Z^d$. then we have divergent subsequences $\{n_k\}_{k=1}^\infty$ and $\{m_k\}_{k=1}^\infty$ of $\N$ such that $\mu_{\Phi,\Lambda_{n_k}}^+$ and $\mu_{\Phi,\Lambda_{m_k}}^-$ converges to some $\widehat{\mu}_\Phi^+$ and $ \widehat{\mu}_\Phi^-$ in $M(\Omega)$, respectively. Then, as we mentioned above, $\widehat{\mu}_\Phi^+$ and $\widehat{\mu}_\Phi^-$ are in $K_\Phi$.
		By Proposition \ref{lowtempmain}, $\widehat{\mu}_\Phi^+$ satisfies
		\begin{align}\label{muhat}
			\widehat{\mu}_\Phi^+\left(\left\{\omega\in\Omega\left|\omega(x)=1\right.\right\}\right)=&\lim_{k\to\infty}\mu_{\Phi,\Lambda_{n_k}}^+(\{\omega\in\Omega_{\Lambda_{n_k}}^+\left|\omega(x)=1\right.\})\nonumber\\
			\geq&\ 1-\varepsilon(L)
		\end{align}
		for each $x\in\Z^d$. For each $N\in\N$, we define $\nu_{\Phi,N}^+\in M(\Omega)$ by
		$$
		\nu_{\Phi,N}^+=\frac{1}{|B(N)|}\sum_{x\in B(N)}\tau^x_*\widehat{\mu}_\Phi^+.
		$$
		Since $\Phi$ is translation invariant and $\widehat{\mu}_\Phi^+\in K_\Phi$, $\tau^x_*\widehat{\mu}_\Phi^+\in K_\Phi$ for each $x\in B(N)$. Hence, by the convexity of $K_\Phi$, $\nu_{\Phi,n}^+\in K_\Phi$. Moreover, from Inequality (\ref{muhat}), we have
		\begin{align}\label{nuplus}
			\nu_{\Phi,N}^+\left(\left\{\omega\in\Omega\left|\omega(0)=1\right.\right\}\right)=&\frac{1}{|B(N)|}\sum_{x\in B(N)}\widehat{\mu}_\Phi^+\left(\left\{\omega\in\Omega\left|\omega(x)=1\right.\right\}\right)\nonumber\\
			\geq&\ 1-\varepsilon(L).
		\end{align}
		We can take a divergent subsequence $\{N_l\}_{l=1}^\infty\subset\N$ such that $\nu_{\Phi,N_l}^+$ converges to some $\nu_\Phi^+$ in $M(\Omega)$. It is seen that $\nu_\Phi^+\in I$ and, since $K_\Phi\subset M(\Omega)$ is closed, $\nu_\Phi^+\in K_\Phi$. Hence, by Proposition \ref{GE} we have $\nu_\Phi^+\in I_\Phi$. Moreover, from Inequality (\ref{nuplus}), we have
		\begin{align*}
			\nu_\Phi^+\left(\left\{\omega\in\Omega\left|\omega(0)=1\right.\right\}\right)=&\lim_{l\to\infty}\nu_{\Phi,N_l}^+\left(\left\{\omega\in\Omega\left|\omega(0)=1\right.\right\}\right)\\
			\geq&\ 1-\varepsilon(L).
		\end{align*}
		By doing the same argument to $\widehat{\mu}_\Phi^-$, we have $\nu_\Phi^-\in I_\Phi$ such that $\nu_\Phi^-\left(\left\{\omega\in\Omega\left|\omega(0)=-1\right.\right\}\right)\geq 1-\varepsilon(L)$. Since $\varepsilon(L)<1/2$, we have $\nu_\Phi^+\neq\nu_\Phi^-$ and complete the proof.
	\end{proof}
	
	\subsection{Contours for a configration and Peierls' argument}
	We will give a proof of Proposition \ref{lowtempmain}. We notice that the following argument is based on {\bf Peierls' argument}, which was used to show the coexistence of equilibrium states for nonperturbed Ising models at low temperature (see \cite[Section 3.7.2, 5.7.4]{FV1}).
	
	We introduce the notion of a {\bf contour} on $\Z^d$. We follow the notion in \cite[Chapter 5]{FV1}. For each $x\in\Z^d$, we write $\I_x$ for the closed unit cube in $\R^d$ centered at $x$:
	$$\I_x=x+\left[-\frac{1}{2},\frac{1}{2}\right]^d,$$
	and, for each $\Lambda\subset\Z^d$, define
	$$\M(\Lambda)=\bigcup_{x\in\Lambda}\I_x\subset\R^d.$$
	\begin{dfn}\label{contour}
		$\gamma\subset\R^d$ is a {\bf contour} if
		$$\gamma=\partial\M(\Lambda)$$
		for some c-connected $\Lambda\Subset\Z^d$ (where $\partial$ denotes the boundary with respect to the usual topology of $\R^d$). It is seen that, for each a contour $\gamma$, $\Lambda\Subset\Z^d$ such that $\gamma=\partial\M(\Lambda)$ is uniquely determined. We call such $\Lambda$ the {\bf interior} of $\gamma$ and write $\interior\gamma$.
	\end{dfn}
	It is seen that a contour $\gamma$ is a connected sum of {\bf plaquettes}, which are $(d-1)$-dimensional faces of $\I_x,x\in\Z^d$, such that $\gamma$ devides $\Z^d$ into two $\|\cdot\|_\infty$-connected sets: the interior and exterior of $\gamma$ (for example, see \cite[Appendix B]{FV1}). We write $|\gamma|$ for the number of plaquettes which are contained in $\gamma$.
	We write $\E=\left\{\{x,y\}\subset\Z^d\left|\|x-y\|_1=1\right.\right\}$ and call an element of $\E$ an {\bf edge} (considering the graph $(\Z^d,\E)$).
	There is the one-to-one mapping between plaquettes and $\E$,
	associating a plaquette with the unique edge crossing it. By this mapping, the plaquettes contained in $\gamma$ correspond to elements of $\left\{\{x,y\}\in\E\left|x\in\interior\gamma,y\in\Z^d\setminus\interior\gamma\right.\right\}$.
	
	For each $\Lambda\Subset\Z^d$, let $\Omega_\Lambda^+=\left\{\omega\in\Omega\left| \omega(x)=1 \text{ on }\Z^d\setminus\Lambda\right.\right\}$.
	We write $\Omega^+=\bigcup_{\Lambda\Subset\Z^d}\Omega_\Lambda^+$.
	For each $\omega\in\Omega^+$, we define
	$\Lambda^-(\omega)=\left\{x\in\Z^d\left|\omega(x)=-1\right.\right\}\Subset\Z^d$ and
	$$\M(\omega)=\M(\Lambda^-(\omega)).$$
	We decompose $\partial\M(\omega)$ into connected components:
	$$\partial\M(\omega)=\bigsqcup_{i=1}^n\gamma_i,$$
	then it is shown that each $\gamma_i\ (i=1,\dots,n)$ is a contour and, by the mapping mentioned above, a plaquette contained in some $\gamma_i$ is associated with an edge $\{x,y\}\in\E$ such that $\omega(x)=1$ and $\omega(y)=-1$ (see \cite[Section 5.7.4]{FV1}). We write
	$$\Gamma(\omega)=\{\gamma_1,\dots,\gamma_n\}.$$
	Let us consider the Ising model with inverse temparature $\beta>0$.
	For $\Lambda\Subset\Z^d$, we define $\E_\Lambda=\left\{\{x,y\}\in\E\right|$
	$\left.\{x,y\}\cap\Lambda\neq\emptyset\right\}$. Then for $\omega\in\Omega_\Lambda^+$, we have
	\begin{align*}
		\exp\left(-\sum_{\Delta\in\mathscr{F}_\Lambda}\Phi_{\Delta}^\beta(\omega)\right)&=\exp\left(\beta\sum_{\{x,y\}\in\E_\Lambda}\omega(x)\omega(y)\right)\\
		&=e^{\beta\left|\E_\Lambda\right|}e^{-2\beta\left|\left\{\{x,y\}\in\E_\Lambda\left|\omega(x)=1,\omega(y)=-1\right.\right\}\right|}\\
		&=e^{\beta\left|\E_\Lambda\right|}\prod_{\gamma\in\Gamma(\omega)}e^{-2\beta|\gamma|}.
	\end{align*}
	Suppose $\Lambda$ is c-connected. For $\omega\in\Omega_\Lambda^+$ and $\gamma\in\Gamma(\omega)$, we define $\omega_\gamma\in\Omega$ by
	$$
	\omega_\gamma(x)=\begin{cases}
		-\omega(x),&x\in\interior\gamma\\ \omega(x),&x\in\Z^d\setminus\interior\gamma.\end{cases}
	$$
	Then, by the assumption that $\Lambda$ is c-connected, $\omega_\gamma\in\Omega_\Lambda^+$ and
	\begin{equation}\label{flipgamma}
		\Gamma(\omega_\gamma)=\Gamma(\omega)\setminus\{\gamma\}.
	\end{equation}
	Hence we have
	$$
	\exp\left(-\sum_{\Delta\in\mathscr{F}_\Lambda}\Phi_{\Delta}^\beta(\omega)\right)=e^{-2\beta|\gamma|}\exp\left(-\sum_{\Delta\in\mathscr{F}_\Lambda}\Phi_{\Delta}^\beta(\omega_\gamma)\right).
	$$
	This is a key equation in Peierls' argument.
	
	Let us consider a perturbation of $\Phi^\beta$ by a spin-flip symmetric interaction $\Psi$ such that $\|\Psi\|<\delta$ and it takes zero on non-$\|\cdot\|_1$-connected sets. Then, since $\omega_\gamma=-\omega$ on $\interior\gamma$ and $\omega_\gamma=\omega$ on $\Z^d\setminus\interior\gamma$, by the spin-flip symmetry of $\Psi$, we have
	\begin{align*}
		\sum_{\Delta\in\mathscr{F}_\Lambda}\Psi_\Delta(\omega)&=\sum_{\substack{\Delta\in\mathscr{F}_\Lambda,\\\Delta\subset\interior\gamma\text{ or }\Delta\cap\interior\gamma=\emptyset}}\Psi_\Delta(\omega)+\sum_{\substack{\Delta\in\mathscr{F}_\Lambda,\\\Delta\not\subset\interior\gamma,\Delta\cap\interior\gamma\neq\emptyset}}\Psi_\Delta(\omega)\\
		&=\sum_{\substack{\Delta\in\mathscr{F}_\Lambda,\\\Delta\subset\interior\gamma\text{ or }\Delta\cap\interior\gamma=\emptyset}}\Psi_\Delta(\omega_\gamma)+\sum_{\substack{\Delta\in\mathscr{F}_\Lambda,\\\Delta\not\subset\interior\gamma,\Delta\cap\interior\gamma\neq\emptyset}}\Psi_\Delta(\omega).
	\end{align*}
	If $\Delta\in\mathscr{F}_\Lambda$ is $\|\cdot\|_1$-connected, $\Delta\not\subset\interior\gamma$ and $\Delta\cap\interior\gamma\neq\emptyset$, then $\Delta$ contains some edge $\{x,y\}\in\E$ associated with some plaquette contained in $\gamma$. Hence we have
	$$\abs[\bigg]{\sum_{\substack{\Delta\in\mathscr{F}_\Lambda,\\\Delta\not\subset\interior\gamma,\Delta\cap\interior\gamma\neq\emptyset}}\Psi_\Delta(\omega)}\leq|\gamma|\|\Psi\|<|\gamma|\delta
	$$
	and the same holds for $\omega_\gamma$. From these, we have
	\begin{align}\label{weight}
		&\exp\left(-\sum_{\Delta\in\mathscr{F}_\Lambda}\left(\Phi^\beta+\Psi\right)_\Delta(\omega)\right)\nonumber\\&=
		e^{-2\beta|\gamma|}\exp\left(-\sum_{\Delta\in\mathscr{F}_\Lambda}\Phi_{\Delta}^\beta(\omega_\gamma)\right)
		\exp\left(-\sum_{\Delta\in\mathscr{F}_\Lambda}\Psi_\Delta(\omega)\right)\nonumber\\
		&=e^{-2\beta|\gamma|}\exp\left(-\sum_{\Delta\in\mathscr{F}_\Lambda}\left(\Phi^\beta+\Psi\right)_\Delta(\omega_\gamma)\right)
		\exp\left(-\sum_{\substack{\Delta\in\mathscr{F}_\Lambda,\\\Delta\not\subset\interior\gamma,\Delta\cap\interior\gamma\neq\emptyset}}\left(\Psi_\Delta(\omega)-\Psi_\Delta(\omega_\gamma)\right)\right)\nonumber\\
		&=e^{-2\beta|\gamma|+\delta\left(\Lambda,\omega,\gamma\right)}\exp\left(-\sum_{\Delta\in\mathscr{F}_\Lambda}\left(\Phi^\beta+\Psi\right)_\Delta(\omega_\gamma)\right)
	\end{align}
	with
	\begin{equation}\label{delta}
		\left|\delta\left(\Lambda,\omega,\gamma\right)\right|<2|\gamma|\delta.
	\end{equation}
	\subsection{Proof of Proposition \ref{lowtempmain}}
	First, we prove the following lemma about the number of contours surrounding a fixed point. The result for $d=2$ is proved and used in \cite[Chapter 6]{Geo11}. We give a proof for the completeness.
	\begin{lem}\label{numberofcontour}
		For any $n\in\N$, let
		$$\Gamma_{n,0}=\left\{\gamma:\text{contour}\left||\gamma|=n,0\in\interior\gamma\right.\right\}.$$
		Then we have
		$$\left|\Gamma_{n,0}\right|\leq(n+1)C_d^{2n+1},$$
		where $C_d\in\N$ is the number of  ($d-1$)-dimensional faces of $\I_x,x\in\Z^d$ which are not $\{1/2\}\times[-1/2,1/2]^{d-1}$ and intersect with it.
	\end{lem}
	\begin{proof}[{\rm Proof}]
		Let $\gamma\in\Gamma_{n,0}$. Then it is easily seen that, if we write $\mathscr{H}_i=\{1/2+i\}\times[-1/2,1/2]^{d-1}$: the ($d-1$)-dimensional faces for $i\in\Z$, then $\mathscr{H}_i\subset\gamma$ for some $i=0,\dots,n$. We define the graph $G=(V,E)$, where the set of vertices $V$ is the set of all ($d-1$)-dimensional faces of $\I_x,x\in\Z^d$ and the set of edges is the set of all pairs $\{\mathscr{H},\mathscr{H}'\}$ of elements of $V$ such that $\mathscr{H}\cap\mathscr{H}'\neq\emptyset$. For any $\mathscr{H}\in V$, the number of edges which has $\mathscr{H}$ as an end point is $C_d$. By elementary arguments of graph theory, it is seen that, for any finite connected subset $S\subset V$ and any $\mathscr{H}\in S$, there exists a path $p$ in $G$ such that $p$ starts from and ends at $\mathscr{H}$, passes through every vertex in $S$ and does not pass through any other vertices and the length of $p$ is bounded by $2|S|$. We fix $i=0,\dots,n$. Then, by regarding a contour as a finite connected subset in $V$, we can associate each $\gamma\in\Gamma_{n,0}$ such that $\mathscr{H}_i\subset\gamma$ with a path $p$ in $G$ starting form and ends at $\mathscr{H}_i$ which satisfies the conditions for $\gamma$ as above. Since the number of paths starting from $\mathscr{H}_i$ the length of which is bounded by $2n$ is bounded by $C_d^{2n+1}$, we obtain $\left|\Gamma_{n,0}\right|\leq(n+1)C_d^{2n+1}$.
	\end{proof}
	
	The next lemma is a key to prove Proposition \ref{lowtempmain}.
	
	\begin{lem}\label{keylemma}
		Let $\beta>0,0<\delta<1$ and $\Psi$ be a tranlation invariant and spin-flip symmetric interaction which is zero on non-$\|\cdot\|_1$-connected sets and $\|\Psi\|<\delta$. We write $\Phi=\Phi^\beta+\Psi$.
		Let $\Lambda\Subset\Z^d$ be c-connected and $\gamma_1,\dots,\gamma_n$ be contours such that $\interior\gamma_1,\dots,\interior\gamma_n\subset\Lambda$ and they are pairwise disjoint. Then we have
		$$\mu_{\Phi,\Lambda}^+\left(\left\{\omega\in\Omega_{\Lambda}^+\left|\gamma_1,\dots,\gamma_n\in\Gamma(\omega)\right.\right\}\right)\leq\exp\left(-2(\beta-\delta)(|\gamma_1|+\dots+|\gamma_n|)\right)$$
		(where we canonically regard $\mu_{\Phi,\Lambda}^+$ as a measure on $\Omega_\Lambda^+$).
	\end{lem}
	\begin{proof}[{\rm Proof}]
		For each $i=1,\dots,n$, we define a map $\theta_{\gamma_i}:\left\{\omega\in\Omega_\Lambda^+\left|\gamma_i\in\Gamma(\omega)\right.\right\}\ni\omega\mapsto\omega_{\gamma_i}\in\left\{\omega\in\Omega_\Lambda^+\left|\gamma_i\notin\Gamma(\omega)\right.\right\}$. It is clear that $\theta_{\gamma_i}$ is injective. From Equation (\ref{flipgamma}), we can define the composition $\theta=\theta_{\gamma_n}\circ\cdots\circ\theta_{\gamma_1}:\left\{\omega\in\Omega_{\Lambda}^+\left|\gamma_1,\dots,\gamma_n\in\Gamma(\omega)\right.\right\}\to\Omega_\Lambda^+$ and it is injective. Moreover, for each $\omega\in\Omega_\Lambda^+$ such that $\gamma_1,\dots,\gamma_n\in\Gamma(\omega)$, from Equation (\ref{weight}) and Inequality (\ref{delta}), we have
		\begin{align*}
			\exp\left(-\sum_{\Delta\in\F_\Lambda}\Phi_\Delta(\omega)\right)\leq&\  e^{-2(\beta-\delta)|\gamma_1|}\exp\left(-\sum_{\Delta\in\F_\Lambda}\Phi_\Delta(\theta_{\gamma_1}(\omega))\right)\\
			\leq&\ e^{-2(\beta-\delta)|\gamma_1|}e^{-2(\beta-\delta)|\gamma_2|}\exp\left(-\sum_{\Delta\in\F_\Lambda}\Phi_\Delta\left(\theta_{\gamma_2}\left(\theta_{\gamma_1}(\omega)\right)\right)\right)
		\end{align*}
		\begin{align*}
			\qquad\qquad\qquad\qquad\qquad\qquad&\vdots\\
			\leq&\ \exp(-2(\beta-\delta)(|\gamma_1|+\cdots+|\gamma_n|))\exp\left(-\sum_{\Delta\in\F_\Lambda}\Phi_\Delta\left(\theta(\omega)\right)\right).
		\end{align*}
		Hence we have
		\begin{align*}
			&\mu_{\Phi,\Lambda}^+\left(\left\{\omega\in\Omega_{\Lambda}^+\left|\gamma_1,\dots,\gamma_n\in\Gamma(\omega)\right.\right\}\right)\qquad\qquad\qquad\qquad\qquad\qquad\qquad\qquad\qquad\\
			=&\ \frac{1}{Z_{\Phi,\Lambda}^+}\sum_{\substack{\omega\in\Omega_{\Lambda}^+,\\\gamma_1,\dots,\gamma_n\in\Gamma(\omega)}}\exp\left(-\sum_{\Delta\in\mathscr{F}_\Lambda}\Phi_\Delta(\omega)\right)\\
			\leq&\ \exp(-2(\beta-\delta)(|\gamma_1|+\cdots+|\gamma_n|))\cdot\frac{1}{Z_{\Phi,\Lambda}^+}\sum_{\substack{\omega\in\Omega_{\Lambda}^+,\\\gamma_1,\dots,\gamma_n\in\Gamma(\omega)}}\exp\left(-\sum_{\Delta\in\F_\Lambda}\Phi_\Delta\left(\theta(\omega)\right)\right)\\
			\leq&\ \exp(-2(\beta-\delta)(|\gamma_1|+\cdots+|\gamma_n|))
		\end{align*}
		and obtain the statement.
	\end{proof}
	
	We prove Proposition \ref{lowtempmain} using these lemmas.
	
	\begin{proof}[{\rm Proof of Proposition {\rm \ref{lowtempmain}}}]
		Let $\beta,\delta>0$ and $\Psi$ be a tranlation invariant and spin-flip symmetric interaction which is zero on non-$\|\cdot\|_1$-connected sets and $\|\Psi\|<\delta$. We write $\Phi=\Phi^\beta+\Psi$.
		We take arbitrary $\Lambda\Subset\Z^d$ and $x\in\Lambda$ and write $\Omega_{\Lambda,x}^{+,-}=\left\{\omega\in\Omega_\Lambda^+\left|\omega(x)=-1\right.\right\}$. Let $\Gamma_x=\left\{\gamma:\text{ contour}\left|x\in\interior\gamma\right.\right\}$.
		For each $\omega\in\Omega_{\Lambda,x}^{+,-}$, it is seen that $\Gamma(\omega)\cap\Gamma_x\neq\emptyset$. Hence, by Lemmas \ref{keylemma} and \ref{numberofcontour}, we have
		\begin{align*}
			\mu_{\Phi,\Lambda}^+\left(\Omega_{\Lambda,x}^{+,-}\right)\leq&\ \sum_{\gamma\in\Gamma_x}\mu_{\Phi,\Lambda}^+(\{\omega\in\Omega_\Lambda^+\left|\gamma\in\Gamma(\omega)\right.\})\\
			\leq&\ \sum_{\gamma\in\Gamma_x}e^{-2(\beta-\delta)|\gamma|}\\
			=&\ \sum_{k=1}^\infty\sum_{\gamma\in\Gamma_{k,0}}e^{-2(\beta-\delta)k}\\
			\leq&\ \sum_{k=1}^\infty(k+1)C_d^{2k+1}e^{-2(\beta-\delta)k}.
		\end{align*}
		uniformly in $\Lambda$ and $x$. Hence, if we take $\log C_d<L<\infty$ sufficiently large and $\beta-\delta>L$ we have
		\begin{align*}
			\mu_{\Phi,\Lambda}^+\left(\Omega_{\Lambda,x}^{+,-}\right)\leq&\ \sum_{k=1}^\infty(k+1)C_d^{2k+1}e^{-2(\beta-\delta)k}\\
			\leq&\ C_d\sum_{k=1}^\infty(k+1)e^{-2(L-\log C_d)k}\\
			=&\ \varepsilon(L)<\ \frac{1}{2}
		\end{align*}
		and $\varepsilon(L)\to 0$ as $L\to\infty$. We obtain the statement in the case of $+$.
		By the same argument in the case of $-$, we obtain Proposition \ref{lowtempmain}.
	\end{proof}
	\begin{rem}\label{magnetization}
		In the proof of Theorem \ref{mainthm} (2) above, we showed that there exist equilibrium states $\nu_\Phi^+$ and $\nu_\Phi^-$ for $\Phi=\Phi^\beta+\Psi$ such that the magnetization is positive and negative, respectively, that is,
		$$
	\int_\Omega\omega(0)\ d\nu_\Phi^+(\omega)>0\quad\text{and}\quad
		\int_\Omega\omega(0)\ d\nu_\Phi^-(\omega)<0.
		$$
		It can be seen that some condition on order of decay of $\Psi$, like $d$-th order decaying condition, is necessary for $\Phi$ to have non-magnetization-vanishing equilibrium states. In \cite{Ent81}, a translation invariant and spin-flip symmetric two body interaction $\Phi$ is constructed such that every equilibrium state $\mu$ for $\Phi=\Phi^\beta+\Psi$ vanishes magnetization. This interaction $\Psi$ is antiferromagnetic type (that is, $\Psi_{\{x,y\}}(\omega)=K_{x,y}\omega(x)\omega(y),K_{x,y}\geq0$) and small with respect to $\|\cdot\|$, but not $d$-th order decreasing. Moreover, it is shown in \cite{BCK07} that, if $d<s\leq d+1$, then, for  the antiferromagnetic interaction $\Psi$ of the form $K_{x,y}=\delta/|x-y|^s,\delta>0$ (where $|\cdot|$ is the Euclidean distance), every equilibrium state for $\Phi=\Phi^\beta+\Psi$ vanishes magnetization. (It seems unknown whether the uniqueness of the equilibrium state holds for these perturbations.)
	\end{rem}
	
	\section{Further questions}\label{furtherquestions}
	Here we notice some further questions next to our results. We showed that there are more than one equilibrium state for $\Phi=\Phi^\beta+\Psi$, where $\beta>0$ is sufficiently large and $\Psi$ is a translation invariant and spin-flip symmetric interaction such that $\opnorm{\Psi}$ is sufficiently small. When $\Psi=0$, that is, $\Phi=\Phi^\beta$, the nonperturbed Ising model at sufficiently low temperature, the structure of $I_{\Phi^\beta}$ is completely known. Remember $\mu_{\Phi^\beta,\Lambda}^+$ and $\mu_{\Phi^\beta,\Lambda}^-$ defined for each $\Lambda\Subset\Z^d$ in (\ref{plusminusgibbs}). 
	\begin{prop}[Completeness of phase diagram]\label{compphasediagram}
		Let $\beta>\beta_c$.
		We take a sequence $\{\Lambda_n\}_{n=1}^\infty$ of finite and c-connected subsets of $\Z^d$ such that $\Lambda_n\uparrow\Z^d$. Then $\{\mu_{\Phi^\beta,\Lambda_n}^+\}_{n=1}^\infty$ and $\{\mu_{\Phi^\beta,\Lambda_n}^-\}_{n=1}^\infty$ converge to some $\mu_{\Phi^\beta}^+$ and $\mu_{\Phi^\beta}^-$ in $M(\Omega)$, respectively and they are independent of the choice of $\{\Lambda_n\}_{n=1}^\infty$. $\mu_{\Phi^\beta}^+$ and $\mu_{\Phi^\beta}^-$ are translation invariant Gibbs states for $\Phi^\beta$ and, then $\mu_{\Phi^\beta}^+,\mu_{\Phi^\beta}^-\in I_{\Phi^\beta}$. Furthermore, $\mu_{\Phi^\beta}^+$ and $\mu_{\Phi^\beta}^-$ span $I_{\Phi^\beta}$, that is, 
		$$
		I_{\Phi^\beta}=\left\{p\mu_{\Phi^\beta}^++(1-p)\mu_{\Phi^\beta}^-\left|0\leq p\leq1\right.\right\}.
		$$
	\end{prop}
	For a proof of the first half of the statement, see \cite[Chapter 3]{FV1} and for that of the second statement, see \cite{Aiz80} or \cite{Hig81} for $d=2$ and \cite{Bod06} for $d\geq3$. These results are due to the FKG inequality, the crucial property of ferromagnetic interactions. If $\beta>0$ is sufficinetly large and a perturbation interaction $\Psi$ is finite lange, then, as we saw in Section \ref{pirogovsinai}, we can apply the Pirogov-Sinai theory to $\Phi=\Phi^\beta+\Psi$ and the same statement holds (see \cite{Zah84}).
	\begin{prob}For sufficiently large $\beta>0$, determine whether the same statement as Proposition {\rm \ref{compphasediagram}} holds or not in general when $\Psi$ is a translation invariant and spin-flip symmetric interaction such that $\opnorm{\Psi}$ is sufficiently small.
	\end{prob}
	We notice that these general cases include many cases to which the FKG inequality and existing extensions of the Pirogov-Sinai theory can not be applicable.
	
	We notice another question. In Theorem \ref{mainthm}, we saw the stability of uniquness of the equilibrium state for sufficiently small $\beta$ and coexistence for sufficiently large $\beta$.
	Then what is about the case when $\beta>0$ is intermediate? At $\beta=\beta_c$, it is clear that the stability is broken. Hence we have a problrm as follows.
	\begin{prob}
		\begin{enumerate}
			\renewcommand{\labelenumi}{\rm{(\arabic{enumi})}}
			\item Determine for each $0<\beta<\beta_c$ whether the stability of uniqueness of the equilibrium state as Theorem {\rm \ref{mainthm} (1)} holds or not.
			\item Determine for each $\beta_c<\beta<\infty$ whether the stability of coexistence of equilibrium states as Theorem {\rm \ref{mainthm} (2)} holds or not.
		\end{enumerate}
	\end{prob}
	
	\section*{Appendix}
	Here we give a proof of Proposition \ref{perturbcont}.
	\begin{proof}[{\rm Proof of Proposition \ref{perturbcont}}]
		Let $\Phi_0$ be an interaction.
		We take an open set $I_{\Phi_0}\subset U\subset I$. For an interaction $\Phi$, we define the function $F_\Phi:I\to\R$ as 
		$$
		F_\Phi(\mu)=h(\mu)+\int_\Omega A_\Phi d\mu,\quad\mu\in I.
		$$
		Then $F_\Phi$ achieves $P^\Phi=\sup_{\mu\in I}F_\Phi(\mu)$ exactly on $I_\Phi$. Moreover, since $I\ni\mu\mapsto\int_\Omega A_\Phi d\mu\in\R$ is continuous and $I\ni\mu\mapsto h(\mu)\in\R$ is upper semicontinuous, $F_\Phi(\mu)$ is upper semicontinuous in $\mu\in I$. Hence, for $\Phi=\Phi_0$, we have
		$$
		I\setminus U\subset\left\{\mu\in I\left| F_{\Phi_0}(\mu)<P^{\Phi_0}\right.\right\}
		$$
		and the left-hand side is compact and the right-hand side is open in $I$. Then, using the upper semicontinuity of $F_{\Phi_0}$ again, we have $0<\varepsilon<1$ such that
		$$
		I\setminus U\subset\left\{\mu\in I\left| F_{\Phi_0}(\mu)<P^{\Phi_0}-\varepsilon\right.\right\}.
		$$
		We set $\delta=\varepsilon/4$ and take an arbitrary interaction $\Psi$ with $\|\Psi\|<\delta$. We write $\Phi=\Phi_0+\Psi$.
		Then, for any $\mu\in I$, we have
		\begin{align*}
			\left|F_\Phi(\mu)-F_{\Phi_0}(\mu)\right|=&\left|\int_\Omega A_\Phi d\mu-\int_\Omega A_{\Phi_0}d\mu\right|\\
			\leq&\int_\Omega\left|A_\Phi-A_{\Phi_0}\right|d\mu\\
			\leq&\ \|\Psi\|\\
			<&\ \delta
		\end{align*}
		and 
		$$
		\left|P^\Phi-P^{\Phi_0}\right|=\left|\sup_{\mu\in I}F_\Phi(\mu)-\sup_{\mu\in I}F_{\Phi_0}(\mu)\right|\leq\delta.
		$$
		If $\mu\in I\setminus U$, then $F_{\Phi_0}(\mu)<P^{\Phi_0}-\varepsilon$. Hence we have
		\begin{align*}
			F_\Phi(\mu)\leq&\ F_{\Phi_0}(\mu)+\delta\\
			<&\ P^{\Phi_0}-\varepsilon+\delta\\
			\leq&\ P^\Phi-\varepsilon+2\delta\\
			=&\ P^\Phi-\frac{\varepsilon}{2}
		\end{align*}
		and $\mu$ does not achieve the supremum. Hence $\mu\notin I_\Phi$.
	\end{proof}

\ \\
Department of Mathematics\\
Faculty of Science\\
Kyoto University\\
Kitashirakawa Oiwake-cho, Sakyo-ku,\\
Kyoto 606-8502, Japan\\
e-mail: s.usuki@math.kyoto-u.ac.jp
\end{document}